\documentclass{llncs}

\usepackage{proof}

\usepackage{diagrams}

 \usepackage{tikz}
 \usetikzlibrary{trees}

 \usepackage{amssymb}
 \usepackage{amsfonts,amsmath,latexsym}

 \newcommand{\At}{\mathrm{\textrm{At}}}

 \newcommand{\ls}{\mathcal{L}_{\Sigma}}

 \author{Ekaterina Komendantskaya\inst{1} \and John Power\inst{2}}
  \institute{Department of Computing, University of Dundee, UK\thanks{Ekaterina Komendantskaya would like to acknowledge
the support of EPSRC Grant EP/K031864/1.}
\and
Department of Computer Science,
 University of Bath, UK \thanks{John Power would like to acknowledge the support of EPSRC grant EP/K028243/1. No data was generated in the course of
this research.}
 }

 \title{Category theoretic semantics for theorem proving in logic programming: embracing the laxness}

\begin{document}

 \maketitle

 \begin{abstract}
A propositional logic program $P$ may be identified with a $P_fP_f$-coalgebra
on the set of atomic propositions in the program. The corresponding $C(P_fP_f)$-coalgebra,
where $C(P_fP_f)$ is the cofree comonad on $P_fP_f$, describes derivations by resolution.
Using lax semantics, that correspondence may be extended to a class
of first-order logic programs without existential variables. 
The resulting extension captures the proofs by term-matching resolution in logic programming.
Refining the lax approach, we further extend it to arbitrary logic programs.
We also exhibit a refinement of Bonchi and Zanasi's saturation semantics for logic programming that
complements lax semantics.

   \keywordname{ Logic programming, 
 coalgebra, term-matching resolution, coinductive derivation tree,
 Lawvere theories, lax 
     transformations, Kan extensions.}
 \end{abstract}

 \section{Introduction}

 Consider the following two logic programs.

 \begin{example} \label{ex:listnat}
 ListNat (for lists of natural numbers) denotes the logic program\\
 $1.\  \mathtt{nat(0)}  \gets $\\
 $2.\  \mathtt{nat(s(x))} \gets  \mathtt{nat(x)}$\\
$3.\  \mathtt{list(nil)}  \gets $ \\
$4.\  \mathtt{list(cons (x, y))}  \gets  \mathtt{nat(x), list(y)}$\\
 \end{example}

\begin{example} \label{ex:lp}
 GC (for graph connectivity) denotes the logic program\\
 $0.\  \mathtt{connected(x,x)}  \gets $\\
 $1.\  \mathtt{connected(x,y)}  \gets  \mathtt{edge(x,z)}, \mathtt{connected(z,y)}$\\
\end{example}
A critical difference between ListNat and GC is
that in the latter, which is a leading example in Sterling and Shapiro's book~\cite{SS},  
there is a variable $z$ in the tail of the second clause that does not appear in its head. 
The category theoretic consequences of that fact are the central concern of this paper.

It has long been observed, e.g., in~\cite{BM,CLM}, that logic programs induce coalgebras, allowing coalgebraic
modelling of
 their operational semantics. 
  In~\cite{KMP}, we developed the idea for variable-free logic
 programs as follows. Using the definition of a logic program~\cite{Llo}, 
given a set of atoms $At$, one can identify a variable-free logic program
 $P$ built over $At$ with a $P_fP_f$-coalgebra structure
 on $At$, where $P_f$ is the finite powerset functor on $Set$: each atom is the
 head of finitely many clauses in $P$, and the body of each
 clause contains finitely many atoms. Our main result was that if
 $C(P_fP_f)$ is the cofree comonad on $P_fP_f$, then, given a logic
 program $P$ qua $P_fP_f$-coalgebra, the corresponding
 $C(P_fP_f)$-coalgebra structure characterises the and-or
 derivation trees generated by $P$, cf~\cite{GC}.


This result has proved to be
stable, not only having been further developed by us~\cite{KoPS,KSH14,JohannKK15,FK15,FKSP16}, but also forming
the basis for Bonchi and Zanasi's saturation semantics for logic programming (LP)~\cite{BZ,BonchiZ15}. 
In Sections~\ref{sec:backr}, \ref{sec:parallel}, we give an updated
account of the work, with updated definitions, proofs and detailed examples, to start our semantic analysis of derivations and proofs in LP.

In~\cite{KoP}, we extended our analysis from variable-free logic programs to
arbitrary logic programs. Following~\cite{ALM,BM,BMR,KP96},
 given a signature $\Sigma$ of function symbols, we let
 $\mathcal{L}_{\Sigma}$ denote the Lawvere theory generated by
 $\Sigma$, and, given a logic program $P$ with function symbols in $\Sigma$, 
we considered the functor category
 $[\mathcal{L}_{\Sigma}^{op},Set]$, extending the set $At$ of atoms in
 a variable-free logic program to the functor from
 $\mathcal{L}_{\Sigma}^{op}$ to $Set$ sending a natural number $n$ to
 the set $At(n)$ of atomic formulae with at most $n$ variables
 generated by the function symbols in $\Sigma$ and the predicate symbols in $P$. We sought
 to model $P$ by a $[\ls^{op},P_fP_f]$-coalgebra
 $p:At\longrightarrow P_fP_fAt$ that, at $n$, takes an atomic formula
 $A(x_1,\ldots ,x_n)$ with at most $n$ variables, considers all 
 substitutions of clauses in $P$  into clauses with variables among $x_1,\ldots ,x_n$ whose head agrees with $A(x_1,\ldots
 ,x_n)$, and gives the set of sets of atomic formulae in antecedents,
 mimicking the construction for variable-free logic
 programs.
 Unfortunately, that idea was too simple.

 Consider the logic program ListNat, i.e., Example~\ref{ex:listnat}. There
 is a map in $\mathcal{L}_{\Sigma}$ of the form $0\rightarrow 1$ that
 models the nullary function symbol $0$. So, naturality of the map
 $p:At\longrightarrow P_fP_fAt$ in $[\mathcal{L}_{\Sigma}^{op},Set]$
 would yield commutativity of the diagram
\begin{diagram}
At(1) & \rTo^{p_1} & P_fP_fAt(1) \\
\dTo<{At(\mathsf{0})} & & \dTo>{P_fP_fAt(\mathsf{0})} \\
At(0) & \rTo_{p_0} & P_fP_fAt(0)
\end{diagram}
But consider $\mathtt{nat(x)}\in At(1)$: there is no clause of the form $\mathtt{nat(x)}\gets \, $ in
 ListNat, so commutativity of the diagram would imply that there
 cannot be a clause in ListNat of the form $\mathtt{nat(0)}\gets \, $
 either, but in fact there is one.

At that point, proposed resolutions diverged: at CALCO in 2011, we proposed one approach using lax transformations~\cite{KoP}, 
then at CALCO 2013, Bonchi and Zanasi proposed another, using saturation semantics~\cite{BZ}, an example of the positive interaction
generated by CALCO! In fact, as we explain in Section~\ref{sec:sat}, 
the two approaches may be seen as complementary rather than as alternatives. First we shall describe our approach.

We followed the standard category theoretic technique of
relaxing the naturality condition on $p$ to a
 subset condition, e.g., as in~\cite{Ben,BKP,HH,Kelly,KP1}, so that, in general, given a map in  $\mathcal{L}_{\Sigma}$ of the form $f:n \rightarrow m$,
the diagram
\begin{diagram}
At(m) & \rTo^{p_m} & P_fP_fAt(m) \\
\dTo<{At(f)} & & \dTo>{P_fP_fAt(f)} \\
At(n) & \rTo_{p_n} & P_fP_fAt(n)
\end{diagram}
 need not commute, but rather the composite via $P_fP_fAt(m)$ need only
 yield a subset of that via $At(n)$. So, for example,
 $p_1(\mathtt{nat(x)})$ could be the empty set while $p_0(\mathtt{nat(0)})$ could be
 non-empty in the semantics for ListNat as required. We extended $Set$ to $Poset$ in order
to express the laxness, and we adopted established category theoretic research on laxness, notably that of~\cite{Kelly},
in order to prove that a cofree comonad exists and behaves as we wished. 

For a representative class of logic programs, the above semantics 
describes derivations arising from restricting the usual SLD-resolution used in LP to  \emph{term-matching resolution}, cf.~\cite{KoP2,KoPS}. 
As transpired in further studies~\cite{FK15,JohannKK15}, this particular restriction to resolution rule captures the theorem-proving aspect of LP as opposed to 
the problem-solving aspect captured by SLD-resolution with unification. 
We explain this idea in Section~\ref{sec:backr}. Derivation trees arising from proofs by term-matching resolution
were called \emph{coiductive trees} in~\cite{KoP2,KoPS} to mark their connection to the coalgebraic semantics.

Categorical semantics introduced in~\cite{KoP} worked well for ListNat, allowing us to model its coinductive trees, as we show in Section~\ref{sec:recall} (It was not
shown explicitly in~\cite{KoP}).
However, it does not work well for GC, the key difference being that, in ListNat, no variable appears
in a tail of a clause  that does not also appear in its head, i.e., clauses in ListNat contain no \emph{existential} variables.
In contrast, although not expressed in these
terms in~\cite{KoP}, we were unable to model the coinductive trees generated by GC because it is an \emph{existential} program, i.e. program containing clauses with existential variables. 
We worked around the problems in~\cite{KoP}, but only inelegantly.


We give an updated account of~\cite{KoP} in Section~\ref{sec:recall}, going beyond~\cite{KoP} to  explain how coinductive trees
for  logic programs without existential variables are modelled, and explaining the difficulty in modelling coinductive semantics for arbitrary logic
programs.
We then devote Section~\ref{sec:derivation} of the paper to resolution of the difficulty, providing lax semantics for coinductive trees
generated by arbitrary logic progams.

In contrast to this, Bonchi and Zanasi, concerned by the complexity involved
with laxness, proposed the use of saturation, following~\cite{BM}, to provide an alternative category theoretic semantics~\cite{BZ,BonchiZ15}.
Saturation is indeed an established and useful construct, as Bonchi and Zanasi emphasised~\cite{BZ,BonchiZ15}, with a venerable
tradition, and, as they say, laxness requires careful calculation, albeit much less so in the setting of posets than that of categories. 
On the other hand, laxness is
a standard part of category theory, one that has been accepted by computer scientists as the need has arisen, e.g.,
by He Jifeng and Tony Hoare to model data refinement~\cite{HH,HH1,KP1,P}. More fundamentally, saturation can be seen as complementary 
to the use of laxness rather than as an alternative to it, as we shall explain in Section~\ref{sec:sat}. 
This reflects the important connection between the theorem proving and problem solving aspects of proof search in LP, as Section~\ref{sec:backr} further explains.
 So we would suggest that
both approaches are of value, with the interaction between them meriting serious consideration.

Saturation inherently yields a particular kind of compositionality, but one loses the tightness of the relationship between semantic
model and operational behaviour. The latter is illustrated by the finiteness of branching in a coinductive tree, in contrast to
the infinity of possible substitutions, which are inherent in saturation. To the extent that it is possible, we would like to recover operational 
semantics from the semantic model, along the lines of~\cite{PP}, requiring maintenance of intensionality where possible. We regard the
distinction between ListNat and GC as a positive feature of lax semantics, as a goal of semantics is to shed
light on the critical issues of programming, relation of existential programs to theorem-proving in LP being one such~\cite{FK15}. So we regard Section~\ref{sec:sat}
as supporting both lax and saturation semantics, the interaction between them shedding light on logic programming.

\section{Background: theorem proving in LP}\label{sec:backr}

 A \emph{signature} $\Sigma$ consists of a set  $\mathcal{F}$ of function
   symbols $f,g, \ldots$ each equipped with an arity.  Nullary (0-ary) function symbols are
 constants. For any set $\mathit{Var}$ of variables,
 the set $Ter(\Sigma)$ of terms over $\Sigma$ is defined
 inductively as usual:
 \begin{itemize}
 \item $x \in Ter(\Sigma)$ for every $x \in \mathit{Var}$.
 \item If $f$ is an n-ary function symbol ($n\geq 0$) and $t_1,\ldots
   ,t_n \in Ter(\Sigma) $, then $f(t_1,\ldots
   ,t_n) \in Ter(\Sigma)$.  
 \end{itemize}


A \emph{substitution} over $\Sigma$ is a total function $\sigma:
\mathit{Var} \to \mathbf{Term}(\Sigma)$. 
Substitutions are extended from variables
to terms as usual: if $t\in \mathbf{Term}(\Sigma)$ and $\sigma$ is a substitution, then the {\em application}
$\sigma(t)$ is a result of applying $\sigma$ to all variables in $t$.
A substitution $\sigma$ is a \emph{unifier}
for $t, u$ if $\sigma(t) = \sigma(u)$, and is a
\emph{matcher} for $t$ against $u$ if $\sigma(t) = u$. 
A substitution $\sigma$ is a {\em most general unifier} ({\em mgu}) for
$t$ and $u$ if it is a unifier for $t$ and
$u$ and is more general than any other such unifier. A {\em most
  general matcher} ({\em mgm}) $\sigma$ for $t$ against $u$ is defined analogously. 

In line with LP tradition~\cite{Llo}, we also take a set $\mathcal{P}$ of predicate symbols each equipped with an arity.
It is possible to define logic programs over terms only, 
  in line with term-rewriting (TRS) tradition~\cite{Terese}, as we do in~\cite{JohannKK15}, but we will follow the usual LP tradition here. 
This gives us the following inductive definitions of the sets of atomic formulae, Horn clauses and logic programs
(we also include the definition of terms here for convenience):

\begin{definition}\label{df:syntax}

\

Terms $Ter \ ::= \ Var \ | \ \mathcal{F}(Ter,..., Ter)$

   Atomic formulae (or atoms) $At \ ::= \ \mathcal{P}(Ter,...,Ter)$

  (Horn) clauses $HC \ ::= \ At \gets At,..., At$
	
	Logic programs $Prog \ ::= HC, ... , HC$
\end{definition}

In what follows, we will use letters $A,B,C,D$, possibly with subscripts, to refer to elements of $At$.


Given a logic program $P$, we may ask whether a certain atom is logically entailed by $P$. E.g., given the program ListNat we may ask whether 
$\mathtt{list(cons(0,nil))}$ is entailed by ListNat. The following rule, which is a restricted form of the famous SLD-resolution, provides a semi-decision procedure 
to derive the entailment:

\begin{definition}[Term-matching (TM) Resolution]\label{def:resolution}
{\small
\[\begin{array}{c}

\infer[]
    {P \vdash [\ ] }
    { } \ \ \ \ \ \ \ \
  \infer[\text{if}~( A \gets A_1, \ldots, A_n) \in P]
    {P \vdash \sigma A }
    { P \vdash \sigma A_1 \quad \cdots \quad P \vdash \sigma A_n } 
  \end{array}
\]}
\end{definition}

In contrast, the SLD-resolution rule could be presented in the following form: 


$$B_1, \ldots , B_j, \ldots , B_n \leadsto_P \sigma B_1, \ldots, \sigma A_1, \ldots, \sigma A_n, \ldots , \sigma B_n ,$$
if $(A \gets A_1, \ldots, A_n) \in P$,  and  $\sigma$ is the mgu of $A$ and $B_{j}$.
The derivation for $A$ succeeds when $A \leadsto_P [\ ]$; we use $\leadsto_P^*$ to denote several steps of SLD-resolution.

At first sight, the difference between TM-resolution and SLD-resolution seems to be that of notation. 
Indeed, both $ListNat \vdash \mathtt{list(cons(0,nil))}$ and\\
$ \mathtt{list(cons(0,nil))} \leadsto^*_{ListNat} [\ ]$ by the above rules (see also Figure~\ref{pic:tree}). However, 
$ListNat \nvdash \mathtt{list(cons(x,y))}$ whereas $ \mathtt{list(cons(x,y))} \leadsto^*_{ListNat} [\ ]$.
And, even more mysteriously, $GC \nvdash \mathtt{connected(x,y)}$ while $\mathtt{connected(x,y)} \leadsto_{GC} [\ ]$.

As it turns out, TM-resolution 
 reflects  the \emph{theorem proving} side of LP: rules of Definition~\ref{def:resolution} can be used to semi-decide whether a given term $t$ is entailed by $P$.
In contrast, SLD-resolution reflects the \emph{problem solving} aspect of LP: using the SLD-resolution rule, one asks whether, for a given $t$, a substitution $\sigma$ 
can be found such that $P \vdash \sigma(t)$.   There is a subtle but important difference between these two aspects of proof search.

For example, when considering the successful derivation $ \mathtt{list(cons(x,y))}$ $ \leadsto^*_{ListNat}  [\ ]$, we assume that $\mathtt{list(cons(x,y))}$ holds only relative to a computed substitution, e.g.
$\mathtt{x \mapsto 0, \ y \mapsto nil}$. 
Of course this distinction is natural from the point of view of theorem proving: $\mathtt{list(cons(x,y))}$ is not a ``theorem" in this generality, 
but its special case,  $\mathtt{list(cons(0,nil))}$, is. Thus, $ListNat \vdash \mathtt{list(cons(0,nil))}$ but $ListNat \nvdash \mathtt{list(cons(x,y))}$ (see also Figure~\ref{pic:tree}).
Similarly, $\mathtt{connected(x,y)} \leadsto_{GC} [\ ]$ should be read as: $\mathtt{connected(x,y)}$ holds relative to the computed substitution $\mathtt{y\mapsto x}$.

According to the soundness and completeness theorems for SLD-resolution~\cite{Llo}, the derivation $\leadsto$ has \emph{existential} meaning, 
i.e. when $\mathtt{list(cons(x,y))} \leadsto^*_{ListNat} [\ ]$, the succeeded goal $\mathtt{list(cons(x,y))}$ is not meant to be read as universally quantified over $\mathtt{x}$ an $\mathtt{y}$.
On the contrary, TM-resolution proves a universal statement. That is, $GC \vdash \mathtt{connected(x,x)}$ reads as:  $\mathtt{connected(x,x)}$ is entailed by GC for any $\mathtt{x}$.


 Much of our recent work has been devoted to formal understanding of the relation between the theorem proving and problem solving 
aspects of LP~\cite{JohannKK15,FK15}. 
The type-theoretic semantics of TM-resolution, given by ``Horn clauses as types, $\lambda$-terms as proofs" is given in~\cite{FK15,FKSP16}. 


Definition~\ref{def:resolution} gives rise to derivation trees. E.g. the derivation (or, equivalently, the proof) for $ListNat \vdash \mathtt{list(cons(0,nil))}$ can be represented by the following derivation tree: 

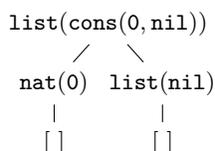
\begin{tikzpicture}[level 1/.style={sibling distance=18mm},
level 2/.style={sibling distance=18mm}, level 3/.style={sibling
distance=18mm},scale=.8,font=\footnotesize,baseline=(current bounding
box.north),grow=down,level distance=10mm]
\node (root) {$\mathtt{list(cons(0,nil))}$}
	child { node {$\mathtt{nat(0)}$}
	child { node {$[\ ]$}
	}}
  	child { node {$\mathtt{list(nil)}$}
		child { node {$[\ ]$}}};
  \end{tikzpicture} 
	
	In general, given a term $t$ and a program $P$, more than one derivation for $P \vdash t$ is possible.
	For example, if we add a fifth clause to program $ListNat$:\\
$5. \  \mathtt{list(cons(0,x)) \gets list(x)}$\\
then yet another, alternative, proof is possible  for the extended program: $ListNat^+ \vdash \mathtt{list(cons(0,nil))}$ via the clause $5$:
	
	
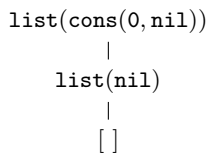
\begin{tikzpicture}[level 1/.style={sibling distance=18mm},
level 2/.style={sibling distance=18mm}, level 3/.style={sibling
distance=18mm},scale=.8,font=\footnotesize,baseline=(current bounding
box.north),grow=down,level distance=10mm]
\node (root) {$\mathtt{list(cons(0,nil))}$}
	  	child { node {$\mathtt{list(nil)}$}
		child { node {$[\ ]$}}};
  \end{tikzpicture}
	
	To reflect the choice of derivation strategies at every stage of the derivation, we can introduce a new kind of nodes, \emph{or-nodes}.  
	For our example, this would give us the tree shown in Figure~\ref{pic:tree}, note the $\bullet$-nodes.
	
	\begin{figure}[t]
\begin{center}
		\begin{tikzpicture}[level 1/.style={sibling distance=30mm},
level 2/.style={sibling distance=20mm}, level 3/.style={sibling
distance=18mm},scale=.8,font=\footnotesize,baseline=(current bounding
box.north),grow=down,level distance=8mm]
\node (root) {$\mathtt{list(cons(0,nil))}$}
  child {[fill] circle (2pt)
	child { node {$\mathtt{nat(0)}$}
	  child {[fill] circle (2pt)
	child { node {$[\ ]$}
	}}}
  	child { node {$\mathtt{list(nil)}$}
		  child {[fill] circle (2pt)
		child { node {$[\ ]$}}}}}
		  child {[fill] circle (2pt)
				child { node {$\mathtt{list(nil)}$}
				  child {[fill] circle (2pt)
		child { node {$[\ ]$}}}}};
  \end{tikzpicture} \ \ \ \ \
		\begin{tikzpicture}[level 1/.style={sibling distance=30mm},
level 2/.style={sibling distance=20mm}, level 3/.style={sibling
distance=18mm},scale=.8,font=\footnotesize,baseline=(current bounding
box.north),grow=down,level distance=8mm]
\node (root) {$\mathtt{list(cons(x,y))}$}
  child {[fill] circle (2pt)
	child { node {$\mathtt{nat(x)}$}}
  	child { node {$\mathtt{list(y)}$}}};
  \end{tikzpicture}
\end{center}
	\caption{\textbf{Left:} a coinductive tree for $\mathtt{list(cons(0,nil))}$ and the extended program $ListNat^+$. \textbf{Right:} a coinductive tree for $\mathtt{list(cons(x,y))}$ and $ListNat^+$. The $\bullet$-nodes mark different clauses applicable to every atom in the tree.}
\label{pic:tree} 
	\end{figure}
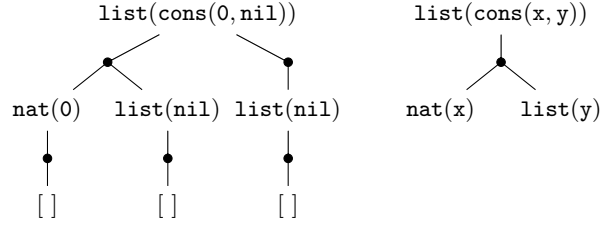
	
	This intuition is made precise in the following definition of a \emph{coinductive tree}, which first appeared in~\cite{KoP,KoPS} and was later refined in~\cite{JohannKK15} under the name of a rewriting tree. Note the use of mgms (rather than mgus) in the last item.

\begin{definition}[Coinductive tree]\label{def:cointree}
Let P be a logic program and $A$ be an atomic formula. The \emph{coinductive
tree} for $A$ is the possibly infinite tree T satisfying the following properties.
\begin{itemize}
\item $A$ is the root of $T$
\item Each node in $T$ is either an and-node or an or-node
\item Each or-node is given by $\bullet$
\item Each and-node is an atom
\item For every and-node $A'$ occurring in $T$, if there exist exactly $m > 0$ distinct clauses
$C_1, \ldots ,C_m$ in P (a clause $C_i$ has the form $B_i \gets  B^i_1,\ldots ,B^i_{n_i}$ for some $n_i$), such that
$A' = B_1\theta_1 = \ldots  = B_m\theta_m$, for mgms $\theta_1,\ldots ,\theta_m$, then $A'$ has exactly $m$ children given
by or-nodes, such that, for every $i \in m$, the $i$-th or-node has $n_i$ children given by and-nodes
$B^i_1\theta_i, \ldots ,B^i_{n_i}\theta_i$.
\end{itemize}
\end{definition}

Coinductive trees provide a convenient model for proofs by TM-resolution.  

Let us make one final observation on TM-resolution.
Generally, given a program $P$ and an atom $t$, one can prove that 

\begin{center}
$ t \leadsto^*_P [\ ]$ with computed substitution $\sigma$ iff  $P \vdash \sigma t$.
\end{center}

This simple fact may give an impression that proofs (and corresponding coinductive trees)  for TM-resolution are in some sense  fragments of reductions by SLD-resolution.
Compare e.g. the right-hand tree of Figure~\ref{pic:tree} before substitution and a grown left-hand tree obtained after the substitution.
 In this case, we could emulate the problem solving aspect of SLD-resolution by using coinductive trees and allowing to apply substitutions within coinductive trees, as was proposed in~\cite{KoP2,JohannKK15,FK15}. 
Such intuition would hold perfectly for e.g. ListNat, but would not hold for existential programs:
although there is a one step SLD-derivation for   $ \mathtt{connected(x,y)} \leadsto_{GC} [\ ]$ (with $\mathtt{y \mapsto x}$), TM-resolution proof for $ \mathtt{connected(x,y)}$ diverges and gives rise to the following infinite
coinductive tree:

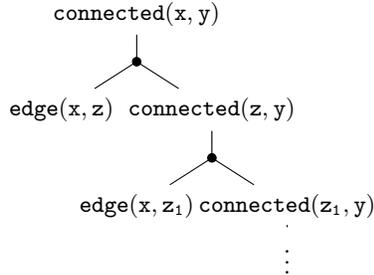
\vspace*{-0.1in}
	\begin{tikzpicture}[level 1/.style={sibling distance=30mm},
level 2/.style={sibling distance=25mm}, level 3/.style={sibling
distance=25mm},scale=.8,font=\footnotesize,baseline=(current bounding
box.north),grow=down,level distance=8mm]
  \node {$\mathtt{connected(x,y)}$}
   child {[fill] circle (2pt)
     child { node {$\mathtt{edge(x,z)}$}}
       child { node {$\mathtt{connected(z,y)}$}
        child {[fill] circle (2pt)
       child { node{$\mathtt{edge(x,z_1)}$}}
         child { node{$\mathtt{connected(z_1,y)}$}
				child {node{$\vdots$}}}}
       }}; 
  \end{tikzpicture}
	
Not only the  proof for $GC \vdash \mathtt{connected(x,y)}$ is not in any sense a fragment of the derivation $\mathtt{connected(x,y)} \leadsto_{GC} [\ ]$, but it also takes 
larger (i.e. infinite) signature.
Thus, operational semantics of TM-resolution and SLD-resolution can be very different for existential programs: both in aspects of termination and signature size. 

This problem is orthogonal to non-termination. Consider  the non-terminating (but not existential) program Bad:

$\mathtt{bad(x)} \gets \mathtt{bad(x)}$\\
For Bad, operational behavior of TM-resolution and SLD-resolution are similar: derivations with both do not terminate and require finite signature.
Once again, such programs can be analysed using similar coinductive methods in TM- and SLD-resolution~\cite{FKSP16,SimonBMG07}.

	The problems caused by existential variables are known in the literature on theorem proving and term-rewriting~\cite{Terese}.
In TRS~\cite{Terese}, existential variables are not allowed to appear in rewriting rules, and in type inference based on term rewriting or TM-resolution, the restriction to non-existential programs is common~\cite{Jones97}. 


So theorem-proving, in contrast to problem-solving,  is modelled by term-matching; term-matching gives
rise to coinductive trees; and as explained in the introduction and, in more detail, later, coinductive trees give rise to laxness. So in this paper, we use laxness to model coinductive trees, and thereby theorem-proving in LP, and relate our semantics with Bonchi and Zanasi's work, which we believe models primarily problem-solving aspect of logic programming.

Categorical semantics for existential programs, which are known to be challenging for theorem proving, is the main contribution of Section~\ref{sec:derivation} and this paper.





 \section{Modelling coinductive trees for variable-free logic programs}\label{sec:parallel}


 In this section, we recall and develop the work of~\cite{KMP} and in particular we
 restrict our semantics to variable-free logic programs, i.e. we take $Var = \emptyset$ in Definition~\ref{df:syntax}.
Variable-free logic programs are operationally equivalent to propositional logic programs, as substitutions 
play no role in derivations. In this (propositional) setting, coinductive trees coincide with the and-or derivation trees known in the LP literature~\cite{GC}.

 \begin{proposition}\label{const:coal}
   For any set $\At$, there is a bijection between the set of
   variable-free logic programs over the set of atoms $\At$ and the set
   of $P_fP_f$-coalgebra structures on $\At$, where $P_f$ is the finite
   powerset functor on $Set$.
 \end{proposition}

 \begin{theorem}\label{constr:Gcoalg}
   Let $C(P_fP_f)$ denote the cofree comonad on $P_fP_f$. Then, for $p:
   \At \longrightarrow P_f P_f(\At)$, the corresponding
   $C(P_fP_f)$-coalgebra $\overline{p}: \At \longrightarrow C(P_fP_f)(\At)$
sends an atom $A$ to the coinductive tree for $A$.
\end{theorem}

\begin{proof}
Applying the work of~\cite{W} to this setting, the cofree comonad is in general determined as follows:
$C(P_fP_f)(\At)$ is the
   limit of the diagram 
 $$\ldots \longrightarrow \At \times P_fP_f(\At
 \times P_fP_f(\At)) \longrightarrow \At \times P_fP_f(\At)
 \longrightarrow \At.$$
with maps determined by the projection $\pi_0:At\times P_fP_f(At)\longrightarrow At$, with applications
of the functor $At \times P_fP_f(-)$ to it.

 Putting $\At_0 = \At$ and $\At_{n+1} = \At \times
 P_fP_f\At_n$, and defining the cone
 \begin{eqnarray*}
   p_0 & = & id: \At \longrightarrow \At ( = \At_0)\\
   p_{n+1} & = & \langle id, P_fP_f(p_n) \circ p \rangle : \At
   \longrightarrow \At \times P_fP_f \At_n ( = \At_{n+1})
 \end{eqnarray*}
 the limiting property of the diagram determines the coalgebra
 $\overline{p}: \At \longrightarrow C(P_fP_f)(\At)$.
The image $\overline{p}(A)$ of an atom $A$ is given by an element of the limit, equivalently a map from $1$ into the limit,
equivalently a cone of the diagram over $1$. 

To give the latter is equivalent to giving an element $A_0$ of $At$, specifically $p_0(A) = A$, together with an element 
$A_1$ of $At\times P_fP_f(At)$, specifically $p_1(A) = (A,p_0(A)) = (A,p(A))$, together with an element 
$A_2$ of $At\times P_fP_f(At\times P_fP_f(At))$, etcetera. The definition of the coinductive tree for $A$ is inherently coinductive, matching
the definition of the limit, and with the first step agreeing with the definition of $p$. Thus it follows by coinduction that $\overline{p}(A)$
can be identified with the coinductive tree for $A$.
 \end{proof}

\begin{example} \label{ex:free}
Let $At$ consist of atoms $\mathtt{A,B,C}$ and $\mathtt{D}$. Let  $P$ denote the logic program
 \begin{eqnarray*}
 \mathtt{A} & \gets & \mathtt{B,C} \\
\mathtt{A} & \gets & \mathtt{B,D} \\
\mathtt{D} & \gets & \mathtt{A,C}\\
 \end{eqnarray*}
So $p(\mathtt{A}) = \{ \{ \mathtt{B,C}\} , \{ \mathtt{B,D} \} \}$, $p(\mathtt{B}) = p(\mathtt{C}) = \emptyset$, and $p(\mathtt{D}) = \{ \{ \mathtt{A,C}\} \}$.

Then $p_0(\mathtt{A}) = \mathtt{A}$, which is the root of the coinductive tree for $\mathtt{A}$.

Then $p_1(\mathtt{A}) = (\mathtt{A},p(\mathtt{A})) = (\mathtt{A},\{ \{ \mathtt{B,C}\} , \{ \mathtt{B,D} \} \})$, which consists of the same information as in the first three levels of the coinductive 
tree for $\mathtt{A}$, i.e., the root $\mathtt{A}$, two or-nodes, and below each of the two
or-nodes, nodes given by each atom in each antecedent of each clause with head $\mathtt{A}$ in the logic program $P$: nodes marked $\mathtt{B}$ and $\mathtt{C}$
lie below the first or-node, and nodes marked $\mathtt{B}$ and $\mathtt{D}$ lie below the second or-node, exactly as $p_1(\mathtt{A})$ describes.

Continuing, note that $p_1(\mathtt{D}) = (\mathtt{D},p(\mathtt{D})) = (\mathtt{D},\{ \{\mathtt{A,C}\} \})$. So
\[
\begin{array}{ccl}
p_2(\mathtt{A}) & = & (\mathtt{A},P_fP_f(p_1)(p(\mathtt{A})))\\
            & = & (\mathtt{A},P_fP_f(p_1)( \{ \{ \mathtt{B,C}\} , \{ \mathtt{B,D} \} \}))\\
            & = & (\mathtt{A}, \{ \{( \mathtt{B},\emptyset),(\mathtt{C},\emptyset)\} , \{( \mathtt{B},\emptyset),(\mathtt{D},\{ \{\mathtt{A,C}\} \}) \} \})
\end{array}
\]
which is the same information as that in the first five levels of the coinductive tree for 
$\mathtt{A}$: $p_1(\mathtt{A})$ provides the first three levels of $p_2(\mathtt{A})$ because $p_2(\mathtt{A})$
must map to $p_1(\mathtt{A})$ in the cone; in the coinductive tree, there are two and-nodes at level 3, labelled by $\mathtt{A}$ and $\mathtt{C}$. As there are no clauses
with head $\mathtt{B}$ or $\mathtt{C}$, no or-nodes lie below the first three of the and-nodes at level 3. However, there is one or-node lying
below $\mathtt{D}$, it branches into and-nodes labelled by $\mathtt{A}$ and $\mathtt{C}$, which is exactly as $p_2(\mathtt{A})$ tells us. 
\end{example}

For pictures of such trees, see~\cite{KoPS}.

 \section{Modelling coinductive trees for logic programs without existential variables}\label{sec:recall}

We now lift the restriction on $Var = \emptyset$ in Definition~\ref{df:syntax}, and consider first-order terms and atoms in full generality, however, we
restrict the definition of clauses in Definition~\ref{df:syntax} to those not containing existential variables.

The \emph{Lawvere theory}  $\mathcal{L}_{\Sigma}$ \emph{generated by} a
signature $\Sigma$ is (up to isomorphism, as there are several equivalent formulations) 
the category defined as follows: $\texttt{ob}(\mathcal{L}_{\Sigma})$ is the set of
   natural numbers.
For each natural number $n$, let $x_1,\ldots ,x_n$ be a specified list
 of distinct variables. Define
 $\mathcal{L}_{\Sigma}(n,m)$ to be the set of $m$-tuples
 $(t_1,\ldots ,t_m)$ of terms generated by the function symbols in
 $\Sigma$ and variables $x_1,\ldots ,x_n$. Define composition in
 $\mathcal{L}_{\Sigma}$ by substitution. 

 One can readily check that these constructions satisfy the axioms for
 a category, with $\mathcal{L}_{\Sigma}$ having
 strictly associative finite products given by the sum of natural
 numbers. The terminal object of $\mathcal{L}_{\Sigma}$ is the natural
 number $0$.

 \begin{example}\label{ex:arrows}
   Consider ListNat. 
   The constants $\mathtt{O}$ and $\mathtt{nil}$ are maps
   from $0$ to $1$ in $\mathcal{L}_{\Sigma}$, $\mathtt{s}$ is modelled by
   a map from $1$ to $1$, and $\mathtt{cons}$ is modelled by a map from
   $2$ to $1$. The term $\mathtt{s(0)}$ is the map
   from $0$ to $1$ given by the composite of the maps modelling
   $\mathtt{s}$ and $\mathtt{0}$. 
 \end{example}

Given an arbitrary logic program $P$ with signature $\Sigma$,
 we can extend the set $At$ of atoms for a variable-free
 logic program to the functor $At:\ls^{op}
 \rightarrow Set$ that sends a natural number $n$ to the set of all
 atomic formulae, with variables among $x_1,\ldots ,x_n$, generated by
 the function symbols in $\Sigma$ and by the predicate symbols in $P$. 
A map $f:n \rightarrow m$ in $\ls$ is sent to
 the function $At(f):At(m) \rightarrow At(n)$ that sends an atomic
 formula $A(x_1, \ldots,x_m)$ to $A(f_1(x_1, \ldots ,x_n)/x_1, \ldots
 ,f_m(x_1, \ldots ,x_n)/x_m)$, i.e., $At(f)$ is defined by
 substitution.

 As explained in the Introduction and in~\cite{KMP}, we cannot model a logic program by a
 natural transformation of the form $p:At\longrightarrow P_fP_fAt$ as
 naturality breaks down, e.g., in ListNat. So, in~\cite{KoP,KoPS}, we relaxed naturality to lax naturality. 
In order to define it, we extended
 $At:\ls^{op}\rightarrow Set$ to have codomain $Poset$ by composing $At$ with the inclusion of $Set$ into
 $Poset$. Mildly overloading notation, we denote the
 composite by $At:\ls^{op}\rightarrow Poset$.
 
 \begin{definition}
   Given functors $H,K:\ls^{op} \longrightarrow Poset$, a {\em lax
     transformation} from $H$ to $K$ is the assignment to each
   object $n$ of $\ls$, of an order-preserving function $\alpha_n: Hn \longrightarrow Kn$ such
   that for each map $f:n \longrightarrow m$ in $\ls$, one has
   $(Kf)(\alpha_m) \leq (\alpha_{n})(Hf)$, pictured as follows:
\begin{diagram}
Hm & \rTo^{\alpha_m} & Km \\
\dTo<{Hf} & \geq & \dTo>{Kf} \\
Hn & \rTo_{\alpha_n} & Kn
\end{diagram}
 \end{definition}
 
 Functors and lax transformations, with pointwise composition, form a locally ordered category
 denoted by $Lax(\ls^{op},Poset)$. Such categories and generalisations have been studied 
extensively, e.g., in~\cite{Ben,BKP,Kelly,KP1}.

 \begin{definition}\label{def:poset}
 Define $P_f:Poset\longrightarrow Poset$ by letting
   $P_f(P)$ be the partial order given by the set of finite subsets of
   $P$, with $A\leq B$ if for all $a \in A$, there exists
   $b \in B$ for which $a\leq b$ in $P$, with behaviour on maps
   given by image. Define $P_c$ similarly but with countability
   replacing finiteness.
 \end{definition}

We are not interested in arbitrary posets in modelling logic programming, only those that arise, albeit inductively, by taking subsets of a set 
qua discrete poset. So we gloss over the fact that, for an arbitrary
poset $P$, Definition~\ref{def:poset} may yield factoring, with the underlying set of $P_f(P)$ being
a quotient of the set of subsets of $P$. It does not affect the line of development here.

\begin{example}\label{ex:listnat2}
Modelling Example~\ref{ex:listnat}, ListNat generates a lax transformation of the form  $p:At\longrightarrow P_fP_fAt$ as follows:
$At(n)$ is the set of atomic formulae in $ListNat$ with at most $n$ variables. 

For example, $At(0)$ consists of $\mathtt{nat(0)}$, $\mathtt{nat(nil)}$, $\mathtt{list(0)}$, $\mathtt{list(nil)}$,
$\mathtt{nat(s(0))}$, $\mathtt{nat(s(nil))}$, $\mathtt{list(s(0))}$, $\mathtt{list(s(nil))}$, $\mathtt{nat(cons(0, 0))}$, $\mathtt{nat(cons( 0, nil))}$, \\
$\mathtt{nat(cons (nil, 0))}$, $\mathtt{nat(cons( nil, nil))}$, etcetera.

Similarly, $At(1)$ includes all atomic formulae containing at most one (specified) variable $x$, thus all the elements of
$At(0)$ together with $\mathtt{nat(x)}$, $\mathtt{list(x)}$, $\mathtt{nat(s(x))}$, $\mathtt{list(s(x))}$, 
$\mathtt{nat(cons( 0, x))}$, $\mathtt{nat(cons (x, 0))}$, $\mathtt{nat(cons (x, x))}$, etcetera.

The function $p_n:At(n)\longrightarrow P_fP_fAt(n)$ sends each element of $At(n)$, i.e., each 
atom $A(x_1,\ldots ,x_n)$ with variables among
$x_1,\ldots ,x_n$, to the set of sets of atoms in the antecedent of each unifying substituted instance of a clause in $P$ with head for which a
unifying substitution agrees with $A(x_1,\ldots ,x_n)$.

Taking $n=0$, $\mathtt{nat(0)}\in At(0)$ is the head of one clause, and there is no other clause for which a unifying substitution will make
its head agree with $\mathtt{nat(0)}$. The clause with head $\mathtt{nat(0)}$ has the empty set of atoms as its tail,  so $p_0(\mathtt{nat(0)}) = \{ \emptyset \}$.

Taking $n=1$, $\mathtt{list(cons( x, 0))}\in At(1)$ is the head of one clause given by a unifying substititution applied to the final clause
of ListNat, and accordingly $p_1(\mathtt{list(cons (x, 0))}) = \{ \{ \mathtt{nat(x)},\mathtt{list(0)} \} \}$.

The family of functions $p_n$ satisfy the inequality required to form a lax transformation precisely because of the allowability
of substitution instances of clauses, as in turn is required to model logic programming. The family does not satisfy the strict
requirement of naturality as explained in the introduction.
\end{example}

\begin{example}\label{ex:lp2}
Attempting to model Example~\ref{ex:lp} by
 mimicking the model of ListNat as a lax transformation of the form $p:At\longrightarrow P_fP_fAt$ in Example~\ref{ex:listnat2} fails.

Consider the clause
\[
 \mathtt{connected(x,y)}  \gets  \mathtt{edge(x,z)},
 \mathtt{connected(z,y)} 
\]

Modulo possible renaming of variables, the head of the clause, i.e., the atom $\mathtt{connected(x,y)}$, lies in $At(2)$ as it has two variables. However,
the tail does not lie in $P_fP_fAt(2)$ as the tail has three variables rather than two. 

We dealt with that inelegantly in~\cite{KoP}: in order to allow $p_2(\mathtt{connected(x,y)})$ to model
GC in any reasonable sense, we allowed substitutions for $z$ by any term on $x,y$ on the basis that
there is no unifying such, so we had better allow all possibilities. So, rather than modelling
the clause directly, recalling that $At(2)\subseteq At(3)\subseteq At(4)$, etcetera, modulo renaming of variables, we put

{\small{
\begin{eqnarray*}
p_2(\mathtt{connected(x,y)}) & = &  \{  \{\mathtt{edge(x,x)},\mathtt{connected(x,y)}\}, \{\mathtt{edge(x,y)},\mathtt{connected(y,y)}\}\}\\
p_3(\mathtt{connected(x,y)}) & = &  \{  \{\mathtt{edge(x,x)},\mathtt{connected(x,y)}\}, \{\mathtt{edge(x,y)},\mathtt{connected(y,y)}\},\\
                                                &     &    \{\mathtt{edge(x,z)},\mathtt{connected(z,y)}\} \}\\
p_4(\mathtt{connected(x,y)}) & = &  \{  \{\mathtt{edge(x,x)},\mathtt{connected(x,y)}\}, \{\mathtt{edge(x,y)},\mathtt{connected(y,y)}\},\\
                                                &     &    \{\mathtt{edge(x,z)},\mathtt{connected(z,y)}\},\{\mathtt{edge(x,w)},\mathtt{connected(w,y)}\} \}
\end{eqnarray*}}}
etcetera: for $p_2$, as only two variables $x$ and $y$ appear in any element of $P_fP_fAt(2)$, we allowed substitution by either 
$x$ or $y$ for $z$; for $p_3$, a third variable may appear in an element of $P_fP_fAt(3)$, allowing an additional possible subsitution; for $p_4$, a fourth variable may appear, etcetera.

Countability arises if a unary symbol $s$ is added to GC, as  in that case, for $p_2$, not only did we allow $x$ and $y$ to be substituted
for $z$, but we also allowed $s^n(x)$ and $s^n(y)$ for any $n>0$, and to do that, we replaced $P_fP_f$ by $P_cP_f$, allowing
for the countably many possible substitutions.

Those were  inelegant
decisions, but they allowed us to give some kind of model of all logic programs.
\end{example}

We now turn to the relationship between the lax transformation
$p:At\longrightarrow P_cP_fAt$ modelling a logic program $P$ and
$\overline{p}:At\longrightarrow C(P_cP_f)At$, the corresponding
coalgebra for the cofree comonad $C(P_cP_f)$ on $P_cP_f$. 

We recall the central abstract result of~\cite{KoP}, the notion of an
"oplax" map of coalgebras being required to match that of lax
transformation. Notation of the form \mbox{$H$-$coalg$} refers to
coalgebras for an endofunctor $H$, while notation of the form
\mbox{$C$-$Coalg$} refers to coalgebras for a comonad $C$. The subscript
$oplax$ refers to oplax maps, and given an
endofunctor $E$ on $Poset$, the notation $Lax(\ls^{op},E)$ denotes the
endofunctor on $Lax(\ls^{op},Poset)$ given by post-composition with
$E$; similarly for a comonad.

 \begin{theorem}\label{main}
 For any locally ordered endofunctor $E$ on $Poset$, if $C(E)$ is
 the cofree comonad on $E$, then there is a canonical isomorphism
 \[
 Lax(\ls^{op},E)\mbox{-}coalg_{oplax} \simeq
 Lax(\ls^{op},C(E))\mbox{-}Coalg_{oplax}
 \]
 \end{theorem}

\begin{corollary}\label{oldcor}
$Lax(\ls^{op},C(P_cP_f))$ is the cofree
comonad on $Lax(\ls^{op},P_cP_f)$.
\end{corollary}
Corollary~\ref{oldcor} gives a bijection between lax transformations
\[
p:At\longrightarrow P_cP_fAt
\]
and lax transformations 
\[
\overline{p}:At\longrightarrow C(P_cP_f)At
\]
subject to the two conditions required of a coalgebra of a comonad. Subject to the routine replacement of the outer
copy of $P_f$ by $P_c$ in the construction in Theorem~\ref{constr:Gcoalg}, the same
construction, if understood pointwise, extends to this setting, i.e., if one uniformly replaces $At$ by $At(n)$ in the construction of
Theorem~\ref{constr:Gcoalg}, and replaces
the outer copy of $P_f$ by $P_c$, one obtains a description of $C(P_cP_f)At(n)$ together with the construction of $\overline{p}_n$ from
$p_n$.

That is fine for ListNat, modelling the coinductive trees generated by ListNat, the same holding for any logic program without existential variables, but for GC, as explained in 
Example~\ref{ex:lp2}, $p$ did \emph{not} model the clause
\[
 \mathtt{connected(x,y)}  \gets  \mathtt{edge(x,z)},
 \mathtt{connected(z,y)} 
\]
directly, and so its extension \emph{a fortiori} could \emph{not} model the 
coinductive trees generated by $\mathtt{connected(x,y)}$.

For arbitrary logic programs, $\overline{p}(A(x_1,\ldots ,x_n))$ was  a variant of the coinductive tree
generated by $A(x_1,\ldots ,x_n)$  in two key ways:
\begin{enumerate}
\item coinductive trees allow new variables to be introduced as one passes down the tree, e.g., with
\[
 \mathtt{connected(x,y)}  \gets  \mathtt{edge(x,z)},
 \mathtt{connected(z,y)}
\]
appearing directly in it, whereas, extending Example~\ref{ex:lp2},  $\overline{p_1}(\mathtt{connected(x,y)})$ does not model such a clause directly, 
but rather substitutes terms on $x$ and $y$ for $z$, continuing inductively as one proceeds.
\item coinductive trees are finitely branching, as one expects in logic programming, whereas $\overline{p}(A(x_1,\ldots ,x_n))$ 
could be infinitely branching, e.g., for GC with an additional unary operation $s$.
\end{enumerate}

 \section{Modelling coinductive trees for arbitrary logic progams}\label{sec:derivation}

We believe that our work in~\cite{KoP} provides an interesting model of ListNat, in particular because it agrees with the coinductive trees
generated by ListNat. However, the account in~\cite{KoP} is less interesting when applied to GC, thus in the full generality
of logic programming. Restriction to non-existential examples such
as ListNat is common for implementational reasons~\cite{KoPS,JohannKK15,FK15,FKSP16}, so~\cite{KoP} does allow the
modeling of coinductive trees for a natural class of logic programs. Here we seek to model coinductive trees for  logic 
programs  in general, \emph{a fortiori} doing so for GC. 

In order to model coinductive trees, it follows from Example~\ref{ex:lp2} that the endofunctor $Lax(\ls^{op},P_fP_f)$ on $Lax(\ls^{op},Poset)$ that sends $At$ to $P_fP_fAt$, needs to be refined as $\{ \{\mathtt{edge(x,z),connected(z,y)}\}\}$ is
not an element of $P_fP_fAt(2)$ as it involves three variables $x$, $y$ and $z$. Motivated by that example, we refine our axiomatics in general so
that the codomain of $p_n$ is a superset of $P_fP_fAt(m)$ for every $m\geq n$. There are six injections of $2$ into $3$, inducing six inclusions $At(2)\subseteq At(3)$, so six inclusions $P_fP_fAt(2)\subseteq P_fP_fAt(3)$, and one only wants to count each element of $P_fP_fAt(2)$ once. So we refine
$P_fP_fAt(n)$ to become $(\Sigma_{m\geq n} P_fP_fAt(m))/\equiv$, where $\equiv$ is generated by the injections $i:n\longrightarrow m$.
This can be made precise in abstract category theoretic terms as follows. 

For any Lawvere theory $L$, there is a canonical identity-on-objects functor from the category $Inj$ of injections $i:n\longrightarrow m$ of natural numbers into $L^{op}$.
So, in particular, there is a canonical identity on objects functor $J:Inj\longrightarrow \ls^{op}$, upon which 
$\Sigma_{m\geq n} P_fP_fAt(m)/\equiv$ may be characterised as the colimit (see~\cite{Mac} or, for the enriched version,~\cite{K})
\[
\int^{m\in n/Inj} P_fP_fAtJ(m)
\]
or equivalently, given $n\in Inj$, the colimit of the functor from
$n/Inj$ to $Poset$ that sends an injection $j:n\longrightarrow m$ to
$P_fP_fAtJ(m)$.

This construction extends to a functor $P_{ff}(At):\ls^{op}\longrightarrow Poset$ by sending a map $f:n\longrightarrow n'$ in $\ls$ to the 
order-preserving function
\[
\int^{m\in n'/Inj} P_fP_fAtJ(m) \longrightarrow \int^{m\in n/Inj} P_fP_fAtJ(m)
\]
determined by the fact that each $m\in n'/Inj$ is, up to coherent isomorphism, uniquely of the form $n'+k$, allowing one to 
apply $P_fP_fAt$ to the map $f+k:n+k\longrightarrow n'+k = m$ in $\ls$.
This is similar to the behaviour of the monad for local state on maps~\cite{PP2}.

It is routine to generalise the construction from $At$ to make it apply to an arbitrary functor $H:\ls^{op}\longrightarrow Poset$. 

In order to make the construction functorial, i.e., in order to make it respect 
maps $\alpha:H\Rightarrow K$, we need to refine 
$Lax(\ls^{op},Poset)$ as the above colimit strictly
respects injections, i.e., 
for any \emph{injection} $i:n\longrightarrow m$, we want the diagram
\begin{diagram}
Hn & \rTo^{\alpha_n} & Kn \\
\dTo<{Hi} & & \dTo>{Ki} \\
Hm & \rTo_{\alpha_m} & Km 
\end{diagram}
to commute.

Summarising this discussion yields the following:

\begin{definition}\label{def:laxinj}
Let $Lax_{Inj}(\ls^{op},Poset)$ denote the category with objects given by functors from $\ls^{op}$ to $Poset$, maps given by lax transformations
that strictly respect injections, and composition given pointwise. 
\end{definition}

\begin{proposition}\label{prop:ff} cf~\cite{PP2}
Let $J:Inj\longrightarrow \ls^{op}$ be the canonical inclusion. 
Define 
\[
P_{ff}:Lax_{Inj}(\ls^{op},Poset)\longrightarrow Lax_{Inj}(\ls^{op},Poset)
\]
by
$(P_{ff}(H))(n) = \int^{m\in n/Inj} P_fP_fHJ(m)$,
with, for any map $f:n\longrightarrow n'$ in $\ls$,
\[
(P_{ff}(H))(f):\int^{m\in n'/Inj} P_fP_fHJ(m) \longrightarrow \int^{m\in n/Inj} P_fP_fHJ(m)
\]
determined by the fact that each $m\in n'/Inj$ is, up to coherent isomorphism, uniquely of the form $n'+k$, allowing one to 
apply $P_fP_fH$ to the map $f+k:n+k\longrightarrow n'+k = m$ in $\ls$.

Given $\alpha:H\Rightarrow K$, define $P_{ff}(\alpha)(n)$ 
by the fact that $m\in n/Inj$ is uniquely of the form $n+k$, and using
\[
\alpha_{n+k}:H(m) = H(n+k)\longrightarrow K(n+k) = K(m)
\]
Then $P_{ff}$ is an endofunctor on  $Lax_{Inj}(\ls^{op},Poset)$.
\end{proposition}

The proof is routine but requires lengthy calculation involving colimits. Observe that we have not required countability anywhere 
in the definition of $P_{ff}$, using only finiteness as we sought
at the end of Section~\ref{sec:recall}.

We can now model an arbitrary logic program by a map $p:At\longrightarrow P_{ff}At$ in $Lax_{Inj}(\ls^{op},Poset)$, 
modelling ListNat as we did in Example~\ref{ex:listnat2} but now modelling the clauses of GC directly rather than
using the awkward substitution instances of Example~\ref{ex:lp2}.

\begin{example}\label{ex:listnat3}
Except for the restriction of $Lax(\ls^{op},Poset)$ to $Lax_{Inj}(\ls^{op},Poset)$, ListNat is modelled in exactly the same
way here as it was in Example~\ref{ex:listnat2}, the reason being that no clause in ListNat has a variable in the tail that
does not already appear in the head. We need only observe that, although $p$ is not strictly natural in general, it
does strictly respect injections. For example, if one views $\mathtt{list(cons( x, 0))}$ as an element of $At(2)$, its image
under $p_2$ agrees with its image under $p_1$.
\end{example}

\begin{example}\label{ex:lp3}
In contrast to Example~\ref{ex:lp2}, using $P_{ff}$, we can emulate the construction of Examples~\ref{ex:listnat2} and~\ref{ex:listnat3} 
for ListNat to
model GC. 

Modulo possible renaming of variables, $\mathtt{connected(x,y)}$ is an element of $At(2)$. The function $p_2$ sends
it to the element $\{ \{ \mathtt{edge(x,z)},\mathtt{connected(z,y)}\}\}$ of $(P_{ff}(At))(2)$. This is possible by taking $n=2$
and $m=3$ in the formula for $P_{ff}(At)$ in  Proposition~\ref{prop:ff}. In contrast, $\{ \{ \mathtt{edge(x,z)},\mathtt{connected(z,y)}\}\}$
 is not an element of $P_fP_fAt(2)$, hence
the failure of Example~\ref{ex:lp2}.

The behaviour of $P_{ff}(At)$ on maps ensures that the 
lax transformation $p$ strictly respects injections. For example, if  $\mathtt{connected(x,y)}$ is seen as an element of $At(3)$, 
the additional variable is treated as a fresh variable $w$, so does not affect the image of  $\mathtt{connected(x,y)}$ under $p_3$.
\end{example}

\begin{theorem}
The functor $P_{ff}:Lax_{Inj}(\ls^{op},Poset)\longrightarrow Lax_{Inj}(\ls^{op},Poset)$ induces a cofree
comonad $C(P_{ff})$ on  $Lax_{Inj}(\ls^{op},Poset)$. Moreover, given a logic progam $P$ qua $P_{ff}$-coalgebra $p:At\longrightarrow P_{ff}(At)$,
the corresponding $C(P_{ff})$-coalgebra $\overline{p}:At\longrightarrow C(P_{ff})(At)$ sends an atom $A(x_1,\ldots ,x_n)\in At(n)$ to
the coinductive tree for $A(x_1,\ldots ,x_n)$.
\end{theorem}

\begin{proof}
The construction of Theorem~\ref{constr:Gcoalg}, subject to mild
rephrasing, continues to work here. Specifically, $(C(P_{ff})At)(n)$
is given by the same limit as in Theorem~\ref{constr:Gcoalg} but with
$At$ replaced by $At(n)$ and with $P_fP_f$ replaced by $P_{ff}$:
products in the category $Lax_{Inj}(\ls^{op},Poset)$ are given
pointwise, so the use of projections is the same; $[Inj,Poset]$ is
locally finitely presentable and $P_{ff}$ is an accessible functor,
allowing us to extend the construction of the cofree comonad pointwise
to $[Inj,Poset]$.  It is routine, albeit tedious, to verify
functoriality of $C(P_{ff})$ with respect to all maps and to verify the
universal property. The construction of $\overline{p}$ is given
pointwise, with it following from its coinductive construction that it
yields the coinductive trees as required.
\end{proof}

 The lax naturality in respect to general maps $f:m\longrightarrow n$  means that a substitution applied
to an atom $A(x_1,\ldots ,x_n)\in At(n)$, i.e., application of the function $At(f)$ to $A(x_1,\ldots ,x_n)$, followed by
application of $\overline{p}$, i.e., taking the coinductive tree for the substituted atom, or application
of the function $(C(P_ff)At)f)$ to the coinductive tree for $A(x_1,\ldots ,x_n)$  potentially yield different
   trees: the former substitutes into $A(x_1,\ldots ,x_n)$, then takes its coinductive tree, while the latter applies
a substitution to each node of the coinductive tree for $A(x_1,\ldots ,x_n)$, then prunes to remove redundant branches.

\begin{example}\label{ex:lp4}
Extending Example~\ref{ex:lp3}, consider $\mathtt{connected(x,y)}\in At(2)$. In expressing GC as a map $p:At\longrightarrow P_{ff}At$
in Example~\ref{ex:lp3}, we put 
\[
p_2(\mathtt{connected(x,y)}) = \{ \{ \mathtt{edge(x,z)},\mathtt{connected(z,y)}\}\}
\] 
Accordingly, $\overline{p}_2(\mathtt{connected(x,y)})$ is the coinductive tree for $\mathtt{connected(x,y)}$, thus the infinite tree generated by repeated application of the same clause modulo renaming of variables.

If we substitute $x$ for $y$ in the coinductive tree, i.e., apply the function $(C(P_{ff})At)(x,x)$ to it (see the definition of $L_{\Sigma}$ at the start
of Section~\ref{sec:recall} and observe that $(x,x)$ is a $2$-tuple of terms
generated trivially by the variable $x$), we obtain the same tree but with $y$ systematically replaced by $x$. However, if we substitute $x$ for $y$
in $\mathtt{connected(x,y)}$, i.e., apply the function $At(x,x)$ to it, we obtain $\mathtt{connected(x,x)}\in At(1)$, whose
coinductive tree has additional branching as the first clause of GC, i.e., $\mathtt{connected(x,x)}\gets \,$ may also be applied.

In contrast to this, we have strict naturality with respect to injections: for example, an injection $i:2\longrightarrow 3$ yields the
function $At(i):At(2)\longrightarrow At(3)$ that, modulo renaming of variables, sends $\mathtt{connected(x,y)}\in At(2)$ to itself
seen as an element of $At(3)$, and the coinductive tree for $\mathtt{connected(x,y)}$ is accordingly also  sent by $(C(P_{ff})At)(i)$ to
itself seen as an element of $(C(P_{ff})At)(3)$.
\end{example}

Example~\ref{ex:lp4} illustrates why, although the condition of strict naturality with respect to injections holds for $P_{ff}$, it does not
hold for $Lax(\ls^{op},P_fP_f)$ in Example~\ref{ex:lp2} as we did not model the clause 
\[
 \mathtt{connected(x,y)}  \gets  \mathtt{edge(x,z)},
 \mathtt{connected(z,y)}
\]
directly there, but rather modelled all substitution instances into all available variables.

\section{Complementing saturated semantics}\label{sec:sat}

Bonchi and Zanasi's approach to modelling logic programming in~\cite{BZ} was to consider $P_fP_f$ as we
did in~\cite{KoP}, sending $At$ to $P_fP_fAt$, but to ignore the inherent laxness, replacing
$Lax(\ls^{op},Poset)$ by $[ob(\ls),Set]$, where $ob(\ls)$ is the set of objects of $\ls$ treated as a discrete
category, i.e., as one with only identity maps.

The central mathematical fact that supports saturated semantics is that, 
regarding $ob(\ls)$ as a discrete category, with inclusion functor $I:ob(\ls)\longrightarrow \ls$, the functor
\[
[I,Set]:[\ls^{op},Set]\longrightarrow [ob(\ls)^{op},Set]
\]
that sends a functor $H:\ls^{op}\longrightarrow Set$ to the composite functor $HI:ob(\ls) = ob(\ls)^{op} \longrightarrow Set$
has a right adjoint. That adjoint is given by right Kan extension. It is primarily the fact of the existence of the right adjoint, rather than its
characterisation as a right Kan extension, that enabled Bonchi and Zanasi's various constructions, in particular those
of saturation and desaturation. 

That allows us to mimic Bonchi and Zanasi's saturation semantics, but starting from $Lax(\ls^{op},Poset)$  rather than from $[ob(\ls),Set]$. We
are keen to allow this as laxness is an inherent fact of the situation, as we have explained through the course of this
paper. Such laxness has been valuable in related semantic endeavours,
such as in Tony Hoare's pioneering work on the modelling of data refinement~\cite{HH,HH1,KP1}, of which substitution in logic programming
can be seen as an instance. 

The argument, which was originally due to Ross Street, cf~\cite{S}, goes as follows.

\begin{theorem}~\cite{BKP}\label{thm:BKP} For any finitary $2$-monad $T$ on a cocomplete $2$-category $K$, the inclusion
\[
J:T\mbox{-}Alg_s\longrightarrow T\mbox{-}Alg_l
\]
of the category of strict $T$-algebras and strict maps of $T$-algebras into the category of strict $T$-algebras and lax
maps of $T$-algebras has a left adjoint.
\end{theorem}

\begin{example}\label{left} For any Lawvere theory $L$, there is a finitary locally ordered monad $T$ on $[ob(L),Poset^{op}]$ for 
which $[L,Poset^{op}]$ is isomorphic
to $T$-$Alg_s$, with $T$-$Alg_l$  isomorphic to $Lax(L,Poset^{op})$. The monad $T$ is given by the composite of the functor
\[
[J,Poset^{op}]:[L,Poset^{op}]\longrightarrow [ob(L),Poset^{op}]
\]
where $J:ob(L)\longrightarrow L$ is the inclusion, cf Bonchi and Zanasi's construction~\cite{BZ}, with its left adjoint, which is given
by left Kan extension. The fact that 
the functor $[J,Poset^{op}]$ also has a right adjoint, given by right Kan extension,
implies that the monad $T$ is finitary.
\end{example}

\begin{corollary}
For any Lawvere theory $L$, the inclusion 
\[
[L^{op},Poset]\longrightarrow Lax(L^{op},Poset)
\]
has a right adjoint.
\end{corollary}

\begin{proof} $Poset$ is a complete $2$-category as it is a complete locally ordered category. So $Poset^{op}$ is a cocomplete $2$-category,
and so $[ob(L),Poset^{op}]$ is a cocomplete $2$-category. So the conditions of Theorem~\ref{thm:BKP} hold for Example~\ref{left}, and
so the inclusion
\[
[L,Poset^{op}]\longrightarrow Lax(L,Poset^{op})
\]
has a left adjoint. But $[L,Poset^{op}]^{op}$ is canonically isomorphic to $[L^{op},Poset]$, and  $Lax(L,Poset^{op})^{op}$ is
canonically isomorphic to $Lax(L^{op},Poset)$, and in general, a functor $H:A\longrightarrow B$ has a right adjoint if and only
if $H:A^{op}\longrightarrow B^{op}$ has a left adjoint. The combination of these facts yields the result.
\end{proof}

With this result in hand, one can systematically work through Bonchi
and Zanasi's paper, adapting their constructions for saturation and
desaturation, without discarding the inherent laxness that logic
programming, cf data refinement, possesses.  

We have stated the results here for arbitrary lax transformations, but
they apply equally to those that strictly respect injections, i.e., a
subtle extension of the above argument shows that the inclusion
\[
[L^{op},Poset]\longrightarrow Lax_{Inj}(L^{op},Poset)
\]
has a right adjoint, that right adjoint being a further variant of the
right Kan extension that Bonchi and Zanasi used. The argument for
lax naturality from the Introduction retains its force, so in Bonchi
and Zanasi's sense, this does not yield compositionality of lax
semantics, but it does further refine their analysis of saturation,
eliminating more double counting.

 \section{Conclusions}\label{sec:concl}
For variable-free logic programs, in~\cite{KMP}, we used the cofree comonad on $P_fP_f$ to model the
coinductive trees generated by a logic program. The notion of coinductive tree had not been isolated
at the time of writing of~\cite{KMP}, or of~\cite{KoP}, so we did not explicitly explain the relationship
in~\cite{KMP}, hence our doing so here, but the result was effectively in~\cite{KMP}, just explained in somewhat
different terms.

Using lax transformations, we extended the result in~\cite{KoP}, albeit again not stating it explicitly but again explained explicitly
here, to arbitrary logic programs, including existential programs 
a leading example being GC, as studied
extensively by Sterling and Shapiro~\cite{SS}.
	The problem of existential clauses is well-known in the literature on theorem proving and within communities that use term-rewriting, TM-resolution or their variants.
In TRS~\cite{Terese}, existential variables are not allowed to appear in rewriting rules, and in type inference, the restriction to non-existential programs is common~\cite{Jones97}. In LP, the problem of handling existential variables when constructing proofs with TM-resolution marks the boundary between the theorem-proving and problem-solving aspects, as explained in Section~\ref{sec:backr}.

The papers~\cite{KoP,KoPS} also contained a kind of category theoretic semantics for 
existential logic programs
such as GC, but that semantics was limited, not modelling the coinductive trees generated by TM-resolution for such logic programs.
Here, we have refined lax semantics, refining $Lax(\ls^{op},Poset)$ to $Lax_{Inj}(\ls^{op},Poset)$, thus  insisting upon
strict naturality for injections, and refining the construction $P_cP_fAt$ to $P_{ff}(At)$, thus allowing for additional variables in the
tail of a clause in a logic program and not introducing countability, cf the modelling of local state in~\cite{PP2}. This has allowed us to model coinductive
trees for arbitrary logic programs.

We have further mildly refined Bonchi and Zanasi's saturation semantics for logic programming~\cite{BZ}, showing how it may be seen to  complement
rather than to replace lax semantics. 



 \bibliographystyle{abbrv}


\end{document}